\documentclass[12pt]{article}

\usepackage{amssymb,amsmath,amsthm,amscd,mathrsfs,amsbsy,amsfonts}
\usepackage{pstricks,pst-node}
\usepackage{amsbsy,young,colortbl,cite}

\usepackage{graphicx}

\usepackage{amsmath}

\usepackage{amssymb,amsthm}
\usepackage[english]{babel}
\newtheorem{lemma}{Lemma}
\theoremstyle{Remark}
\newtheorem{remark}{Remark}

\newtheorem{theorem}{Theorem}

\newcommand{\B}{\mathcal B}

\newcommand{\C}{\mathbb C}
\newcommand{\N}{\mathbb N}

\makeatletter
    
    \newcommand{\Rmnum}[1]{\expandafter\@slowromancap\romannumeral #1@}
  \makeatother

\def\({\left(}
\def\){\right)}
\def\[{\begin{eqnarray}}
\def\]{\end{eqnarray}}

\begin{document}

\title{Rogue waves of the Frobenius nonlinear Schr\"{o}dinger equation }

\author{
\  \ Huijuan Zhou, Chuanzhong Li\dag\footnote{Corresponding author.}\\
\small Department of Mathematics,  Ningbo University, Ningbo, 315211, China\\
\dag Email:lichuanzhong@nbu.edu.cn}

\maketitle

\abstract{
In this paper, by considering the potential application in two mode nonlinear waves in nonlinear fibers under a certain case, we  define a coupled nonlinear Schr\"{o}dinger equation(called  Frobenius NLS equation) including its Lax pair. Afterwards, we construct the Darboux transformations of the Frobenius NLS equation. New solutions can be obtained from known seed solutions by the Dardoux transformations.  Soliton solutions are generated from trivial seed solutions. Also we derive breather solutions $q,r$ of  the Frobenius NLS equation obtained from periodic seed solutions. Interesting enough, we find the amplitudes of $r$  vary in size in different areas with period-like fluctuations in the background.  This is very different from the solution $q$ of the single-component classical nonlinear Schr\"{o}dinger equation. Then, the rogue waves of the Frobenius NLS equation are given explicitly by a Taylor series expansion about the breather solutions $q,r$. Also the graph of rogue wave solution $r$ shows that the rogue wave has fluctuations around the peak. The reason of this phenomenon should be in the dynamical dependence of $r$ on $q$ which is independent on $r$.
}
\\
\\
PACS number(s): 05.45.Yv, 42.65.Tg, 42.65.Sf, 02.30.Ik 

\tableofcontents

\section {Introduction}
\   \   \
In the past several decades, nonlinear science has experienced an explosive growth with  several new exciting and fascinating concepts such as solitons, breathers solution, rogue waves, etc. Most of the completely integrable nonlinear  systems admit one striking nonlinear phenomenon, described as  solitons, which is of great mathematical interest. The study of the solitons and other related solutions has become one of the most exciting and active areas in the field of nonlinear sciences. Among these studies, the nonlinear Schr\"{o}dinger(NLS) equation is one of the most popular soliton equations which arises in different physical context[1]. Such equations have the stable soliton solutions, which are a fine balance between their linear dispersive and nonlinear collapsing terms.

In addition to solitons[2-3], rouge waves become the subject of intensive investigations not only in oceanography[4-6], but also in matter rogue waves[7] in Bose-Einstein condensates, rogue waves in plasmas[8], and financial rogue waves describing the  physical mechanisms in financial markets[9]. Rogue wave is localized in both space and time,
which seems to appear from nowhere and disappear without
a trace. In the above-mentioned fields, soliton systems including the nonlinear Schr\"{o}dinger(NLS) equation[10], derivative NLS systems[11-12] the Hirota equation[13-15], the NLS-Maxwell-Bloch equation[16] and the Hirota and Maxwell-Bloch equation\cite{pre,CPL,sciencephys} are considered and reported to admit rogue wave solutions.

It has been shown that the simple scalar NLS could  describe nonlinear waves in nonlinear fibers well and coupled NLS equations are often
used to describe the interaction among the modes in nonlinear
optics and BEC, etc..  By considering the interaction between two different optical fibers, the symmetric coupled NLS equation was sometimes studied such as \cite{couplednls}. In this article, we will consider another nonsymmetric coupled NLS equation called Frobenius NLS equation as following
\begin{equation}\label{fnlsequation}
\begin{cases}
iq_{t}+q_{xx}+2q^{2}q^{*}=0, \\
ir_{t}+r_{xx}+2q^{2}r^{*}+4qq^{*}r=0.
\end{cases}
\end{equation}

Where  $(q,r)$ represent the wave envelopes, $t$ is
the evolution variable, and $x$ is a spatial independent
variable.
Under the equation \eqref{fnlsequation}, we introduce two mode
optical signals into a nonlinear fiber operating in
the anomalous group velocity dispersion regime, marked
by  $(q,r)$.  This equation should describe the case in nonlinear
optics when one fiber is independent and the amplitude $r$ of nonlinear waves in another fiber can be affected by the amplitude $q$ of this one.

It was remarked that the Darboux transformation is an efficient method to generate soliton solutions of integrable equations[21-22].  The multi-solitons can be obtained by this Darboux transformation from a trivial seed solution. Being different from soliton solutions, breather solutions can be derived from periodic solutions through Darboux transformation. This kind of solution exhibits a good periodic, and other forms of solutions such as rogue wave solutions can be obtained from breather solutions.
In this paper, we shall study the rogue waves of the Frobenius NLS equation by the Darboux transformation.

This paper was arranged as follows: In section 2, we recall some basic facts about the nonlinear Schr\"{o}dinger(NLS) equation and then we give the definition of the Frobenius NLS equation and its Lax representation. In section 3, we calculated the Dardoux transformation of the Frobenius NLS equation. Using Darboux transformation, soliton solutions will be given in section 4 by assuming trivial seed solutions. In section 5, starting from a periodic seed solution, the breather solution of the Frobenius NLS equation is provided. A Taylor expansion from the breather solution will help us to construct the rogue wave solution in section 6.

\section{NLS equation and  Frobenius NLS equation}
\   \
In this section, let us firstly explore the source of the Lax form of the nonlinear Schr\"{o}dinger equation (eq.\eqref{a}). The integrable method of  soliton equations is always to use the zero curvature equation such as following
\begin{equation}\label{hh}
M_{t}-N_{x}+[M,N]=0, \  \  \ [M,N]=MN-NM.
\end{equation}

We can export many soliton equations by choosing appropriate  values of $M$ and $N$ as
\begin{equation}\label{34}
\varphi_{x}=M\varphi,\   \
\varphi_{t}=N\varphi,
\end{equation}
where
\begin{equation}\label{ii}
\varphi=\left(\begin{matrix}
\varphi_{1}\\
\varphi_{2}\end{matrix}
\right), \  \  \
M=\left(\begin{matrix}
-i\lambda&\psi\\
\mu&i\lambda\end{matrix}
\right),  \   \ \N=\left(\begin{matrix}
A&B\\
C&-A\end{matrix}
\right),
\end{equation}
and $\lambda$ is a complex parameter.

Here the zero curvature equation (eq.\eqref{hh}) becomes
\begin{center}
\begin{equation}\label{oo}
\begin{split}
&A_{x}=\psi c-\mu B,\\
&\psi_{t}=B_{x}+2i\lambda B+2\psi A,\\
&\mu_{t}=C_{x}-2i\lambda C-2\mu A.
\end{split}
\end{equation}
\end{center}

For example, by choosing appropriate values of $A$, $B$, $C$ as
\begin{footnotesize}
\begin{equation}\label{1}
\begin{cases}
A=a_{3}^{0}\lambda^{3}+a_{2}^{0}\lambda^{2}+(\frac{a_{3}^{0}}{2}\psi(-\psi^*)+a_{1}^{0})\lambda-\frac{i}{4}a_{3}^{0}(\psi(-\psi_{x}^{*})-(-\psi^*)\psi_{x})+\frac{a_{2}^{0}}{2}\psi(-\psi^*)+a_{0}^{0}, \\B=ia_{3}^{0}\psi\lambda^{2}+(-\frac{a_{3}^{0}}{2}\psi_{x}+ia_{2}^{0}\psi)\lambda+\frac{i}{4}a_{3}^{0}(-\psi_{xx}+2\psi^2(-\psi^*))-\frac{a_{2}^{0}}{2}\psi_{x}+ia_{1}^{0}\psi, \\C=ia_{3}^{0}(-\psi^*)\lambda^{2}+(\frac{a_{3}^{0}}{2}(-\psi_{x}^*)+ia_{2}^{0}(-\psi^*))\lambda+\frac{i}{4}a_{3}^{0}(\psi_{xx}^*+2\psi(-\psi^*)^2)+\frac{a_{2}^{0}}{2}(-\psi_{x}^*)+ia_{1}^{0}(-\psi^*),
\end{cases}
\end{equation}
\end{footnotesize}
where
\begin{equation}\label{jj}
a^{0}_{0}=a^{0}_{1}=a^{0}_{3}=0,\ \ \ a^{0}_{2}=-2i, \  \ \mu=-\psi^*.\notag
\end{equation}
The well-known  nonlinear Schr\"{o}dinger(NLS) equation can be derived as following
 \begin{equation} \label{a}
i\psi_{t}+\psi_{xx}+2\psi^{2}\psi^{*}=0.
\end{equation}

Now one can consider a transformation $\varphi^{[1]}=T\varphi$ which change $\varphi_{x}=M\varphi$ into $\varphi_{x}^{[1]}=M^{[1]}\varphi^{[1]}$ and this transformation can be called Darboux transformation.

 From
 \begin{equation}\label{d}
\varphi^{[1]}_{x}=T_{x}\varphi+T\varphi_{x}=M^{[1]}\varphi^{[1]}=M^{[1]}T\varphi,   \notag
\end{equation}
one can get
\begin{equation}\label{e}
T_{x}+TM=M^{[1]}T,        \notag
\end{equation}
which leads to
\begin{equation}\label{4.4}
M^{[1]}=T_{x}T^{-1}+TMT^{-1}.
\end{equation}

The same procedure may be easily adapted to obtain
\begin{equation}\label{4.7}
N^{[1]}=T_{t}T^{-1}+TNT^{-1}.
\end{equation}

Thus the following identity can be seen
\begin{equation}\label{4.8}
M_{t}^{[1]}-N^{[1]}_{x}+[M^{[1]},N^{[1]}]=T(M_{t}-N_{x}+[M,N])T^{-1}.
\end{equation}

Due to the Darboux transformation, the  zero curvature equation
\begin{equation}\label{g} M_{t}-N_{x}+[M,N]=0,  \notag
 \end{equation}
  will become
 \begin{equation}\label{f}M_{t}^{[1]}-N^{[1]}_{x}+[M^{[1]},N^{[1]}]=0.  \notag
\end{equation}

If $\psi$ is a solution of (eq.\eqref{a}), and under the Dardoux transformation, $\psi^{[1]}$ will be another solution of eq.\eqref{a}
 \begin{equation}\label{4.9}
i\psi_{t}^{[1]}+\psi_{xx}^{[1]}+2\psi^{[1]2}\psi^{[1]*}=0.
\end{equation}

Now let us consider the case when $\psi$ take values in  a specific Frobenius algebra $Z_2=\C[\Gamma]/(\Gamma^2)$ and $\Gamma=(\delta_{i,j+1})_{ij}\in gl(2,\C),$ i.e. $\psi=qE+r\Gamma$,
we can get the the Frobenius NLS equation as
\begin{equation}\label{Z2}
\begin{cases}
iq_{t}+q_{xx}+2q^2q^*=0,  \\
ir_{t}+r_{xx}+2q^2r^*+4qq^{*}r=0,
\end{cases}
\end{equation}
whose Lax equation will be given in the next subsection.
\subsection{Lax equation of the Frobenius NLS equation}
\   \  \
  In the Lax equation of the original NLS equation, we imagine
\begin{equation}\notag
a^{0}_{0}=a^{0}_{1}=a^{0}_{3}=0E, \  \ a^{0}_{2}=-2iE,
\end{equation}
in the eq.\eqref{1}.
Then we can find the Lax equation of the Frobenius NLS equation, i.e.
the following  linear spectral problem of the Frobenius NLS equation
\begin{equation}\label{laxfnls}
\varphi_{x}=F\varphi,\   \
\varphi_{t}=G\varphi.
\end{equation}
Here
\begin{equation}\notag
F=\left(\begin{matrix}
-i\lambda E&qE+r\Gamma\\
-qE-r\Gamma&i\lambda E\end{matrix}
\right),\\
\end{equation}
\begin{equation}\notag
G=\left(\begin{matrix}
(-2i\lambda^{2}+iqq^{*})E+(iq^{*}r+iqr^{*})\Gamma &(2q\lambda+iq_{x})E+(2r\lambda+ir_{x})\Gamma \\
(-2q^{*}\lambda+iq_{x}^{*})E+(-2r^{*}\lambda+ir_{x}^{*})\Gamma &(2i\lambda^{2}-iqq^{*})E-(iq^{*}r+iqr^{*})\Gamma\end{matrix}
\right),\\
\end{equation}
where
\begin{equation}\notag
E=\left(\begin{matrix}
1&0\\
0&1\end{matrix}
\right),
\Gamma=\left(\begin{matrix}
0&0\\
1&0\end{matrix}
\right).
\end{equation}

\section{Darboux transformation of the Frobenius NLS equation}
\   \
In this section,  let us consider the Darboux transformation of the Frobenius NLS equation, i.e. when all functions take values in  the  commutative algebra $Z_2=\C[\Gamma]/(\Gamma^2)$ and $\Gamma=(\delta_{i,j+1})_{ij}\in gl(2,\C).$ we can suppose the  Darboux transformation operator $T$ of the Frobenius NLS equation is as follows
\begin{equation}
T=T_{1}\lambda+T_{0},
\end{equation}
where
\begin{equation}\notag
T_{1}=\left(\begin{matrix}
a_{10}E+a_{11}\Gamma&b_{10}E+b_{11}\Gamma\\
c_{10}E+c_{11}\Gamma&d_{10}E+d_{11}\Gamma\\
\end{matrix}
\right),  \  \
T_{0}=\left(\begin{matrix}
a_{0}E+a_{1}\Gamma&b_{0}E+b_{1}\Gamma\\
c_{0}E+c_{1}\Gamma&d_{0}E+d_{1}\Gamma\\
\end{matrix}
\right).
\end{equation}

For the  eq.\eqref{4.4}, we can get
\begin{quote}
\begin{equation}
\begin{split}
 T_{x}&=T_{1x}\lambda+T_{0x}\\
&=\left(\begin{matrix}
a_{10x}E+a_{11x}\Gamma&b_{10x}E+b_{11x}\Gamma\\
c_{10x}E+c_{11x}\Gamma&d_{10x}E+d_{11x}\Gamma\\
\end{matrix}
\right)\lambda+\left(\begin{matrix}
a_{0x}E+a_{1x}\Gamma&b_{0x}E+b_{1x}\Gamma\\
c_{0x}E+c_{1x}\Gamma&d_{0x}E+d_{1x}\Gamma\\
\end{matrix}
\right)\\ \notag
&=\left(\begin{matrix}
{\bf 0}&-2ib_{10}E-2ib_{11}\Gamma\\
2ic_{10}E+2ic_{11}\Gamma&{\bf 0}\\
\end{matrix}
\right)\lambda^2+\left(\begin{matrix}
H1&H2\\
H3&H4\\
\end{matrix}
\right)\lambda +\left(\begin{matrix}
K1&K2\\
K3&K4\end{matrix}
\right),
\end{split}
\end{equation}
\end{quote}

where
\begin{equation}
 \begin{split}
H_1&=(q^{[1]}c_{10}+b_{10}q^{*})E+(r^{[1]}c_{10}+r^{[1]}c_{11}+b_{10}r^{*}+b_{11}q^{*})\Gamma,\\
H_2&=(-2ib_{0}+q^{[1]}d_{10}-a_{10}q)E+(-2ib_{1}+q^{[1]}d_{11}+r^{[1]}d_{10}-a_{10}r-a_{11}q)\Gamma,\\
H_3&=(2ic_{0}-q^{[1]*}a_{10}+d_{10}q^{*})E+(2ic_{1}-q^{[1]*}a_{11}+r^{[1]*}a_{10}+d_{10}r^{*}+d_{11}q^{*})\Gamma,\\
H_4&=(-q^{[1]*}b_{10}-c_{10}q)E+(-q^{[1]*}b_{11}-r^{[1]*}b_{10}-c_{10}r-c_{11}q)\Gamma,\\
K_1&=(q^{[1]}c_{0}+q^{*}b_{0})E+(r^{[1]}c_{0}+q^{[1]}c_{1}+r^{*}b_{0}+b_{1}q^{*})\Gamma,\\
K_2&=(q^{[1]}d_{0}-a_{0}q)E+(q^{[1]}d_{1}+r^{[1]}d_{0}-a_{0}r-a_{1}q)\Gamma,\\
K_3&=(-q^{[1]*}a_{0}+d_{0}q^{*})E+(-q^{[1]*}a_{1}-r^{[1]*}a_{0}+r^{*}d_{0}+d_{1}q^{*})\Gamma,\\
K_4&=(-q^{[1]*}b_{0}-c_{0}q)E+(-q^{[1]*}b_{1}-r^{[1]*}b_{0}-c_{0}r-c_{1}q)\Gamma.\\
 \end{split}
\end{equation}

By comparing the coefficient in terms of $\lambda$, we can get

\begin{equation}\notag
b_{10}=b_{11}= 0,
\end{equation}
\begin{equation}\notag
c_{10}=c_{11}=0,
\end{equation}
\begin{equation}\notag
a_{10x}=a_{11x}=0,
\end{equation}
\begin{equation}\notag
d_{10x}=d_{11x}=0,
\end{equation}
\begin{equation}\notag
-2ib_{0}+q^{[1]}d_{10}-qa_{10}=0,
\end{equation}
\begin{equation}\notag
-2ib_{1}+r^{[1]}d_{10}+q^{[1]}d_{11}-qa_{11}-ra_{10}=0.
\end{equation}

So functions $a_{10},a_{11}, d_{10},d_{11}$, are independent on $x$. Also  the following
identities must hold
\begin{equation}\notag
a_{0x}=q^{[1]}c_{0}+b_{0}q^{*},
\end{equation}
\begin{equation}\notag
b_{0x}=q^{[1]}d_{0}-a_{0}q,
\end{equation}
\begin{equation}\notag
c_{0x}=-q^{[1]*}a_{0}+d_{0}q^{*},
\end{equation}
\begin{equation}\notag
d_{0x}=-q^{[1]*}b_{0}-c_{0}q,
\end{equation}
\begin{equation}\notag
a_{1x}=r^{[1]}c_{0}+q^{[1]}c_{1}+b_{0}r^{*}+q^{*}b_{1},
\end{equation}
\begin{equation}\notag
b_{1x}=q^{[1]}d_{1}+d_{0}r^{[1]}-a_{0}r-a_{1}q,
\end{equation}
\begin{equation}\notag
c_{1x}=-r^{[1]*}a_{0}-q^{[1]*}a_{1}+d_{0}r^{*}+q^{*}d_{1},
\end{equation}
\begin{equation}\notag
d_{1x}=-q^{[1]*}b_{1}-b_{0}r^{[1]*}-c_{0}r-c_{1}q,
\end{equation}
and we can further get
\begin{equation}
(b_{0}c_{0})_{x}=(a_{0}d_{0})_{x},\  \   \   \
(b_{1}c_{0}+b_{0}c_{0})_{x}=
(a_{1}d_{0}+a_{0}d_{1})_{x}.
\end{equation}

For $T_{t}+TN=N^{[1]}T$, we can get functions $a_{10},a_{11}, d_{10},d_{11}$ should be independent on $t$ using the same method. This suggests that $a_{10}$, $a_{11}$, $d_{10}$ and $d_{11}$ are constants. Then the form of Dardoux transformation operator is as following
\begin{equation}
T=T_{1}(\lambda)
 =\left(\begin{matrix}
a_{10}E+a_{11}\Gamma&{\bf 0}\\
{\bf 0}&d_{10}E+d_{11}\Gamma\\
\end{matrix}
\right)\lambda+\left(\begin{matrix}
a_{0}E+a_{1}\Gamma&b_{0}E+b_{1}\Gamma\\
c_{0}E+c_{1}\Gamma&d_{0}E+d_{1}\Gamma\\
\end{matrix}
\right).\\
\end{equation}

Now we will try to determine the specific expression of the matrix $T$. In order to solve this problem, we must use eigenfunction, which contains $\lambda$, to determine their expressions through the parameterized method. Next, we will introduce the properties of eigenfunctions in the spectral problem in a lemma as follows
\begin{lemma}\label{linearconstraint}
\  \ For  $\varphi_{j10}=\varphi_{j10}(x,t,\lambda_{j})$, $\varphi_{j11}=\varphi_{j11}(x,t,\lambda_{j})$, $\varphi_{j20}=\varphi_{j20}(x,t,\lambda_{j})$, $\varphi_{j21}=\varphi_{j21}(x,t,\lambda_{j})$, $j=1,2,$  if $\varphi_{110}$ , $\varphi_{111}$, $\varphi_{120}$, $\varphi_{121}$ are components of  the solution $\varphi_{1}= \left(\begin{matrix}\varphi_{110}E+\varphi_{111}\Gamma\\
\varphi_{120}E+\varphi_{121}\Gamma\end{matrix}
\right)$  of  eq.\eqref{laxfnls} with  $\lambda=\lambda_{1}.$ And the $-\varphi_{120}^{*}(x,t,\lambda_{1}^{*})$, $-\varphi_{121}^{*}(x,t,\lambda_{1}^{*})$, $\varphi_{110}^{*}(x,t,\lambda_{1}^{*})$, $\varphi_{111}^{*}(x,t,\lambda_{1}^{*})$ will  also be the components of a new solution $\varphi_{2}$ of eq.\eqref{laxfnls}. That means the linear system \eqref{laxfnls} admits the following constraint $ \varphi_{210}=-\varphi_{120}^{*}(x,t,\lambda_{1}^{*})$, $\varphi_{211}=-\varphi_{121}^{*}(x,t,\lambda_{1}^{*})$,
$\varphi_{220}=\varphi_{110}^{*}(x,t,\lambda_{1}^{*})$, $\varphi_{221}=\varphi_{111}^{*}(x,t,\lambda_{1}^{*})$, with $\lambda_{2}=\lambda_{1}^*.$
\end{lemma}

\begin{proof}

If $\varphi_{110}$, $\varphi_{111}$, $\varphi_{120}$, $\varphi_{121}$ is the solution of $eq.\eqref{laxfnls}$, i.e.
\begin{equation}
\begin{split}
\varphi_{110x}&=-i\lambda_{1}\varphi_{110}+q\varphi_{120},\\
\varphi_{111x}&=-i\lambda_{1}\varphi_{111}+q\varphi_{121}+r\varphi_{120},\\
\varphi_{120x}&=i\lambda_{1}\varphi_{120}-q^{*}\varphi_{110},\\
\varphi_{121x}&=i\lambda_{1}\varphi_{121}-q^{*}\varphi_{111}-r^{*}\varphi_{110}.\\
\end{split}
\end{equation}

Take the conjugate for the above equation on the both sides at the same time
\begin{equation}
\begin{split}
\varphi_{110x}^{*}&=i\lambda_{1}^{*}\varphi_{110}^{*}+q^{*}\varphi_{120}^{*},\\
\varphi_{111x}^{*}&=i\lambda_{1}^{*}\varphi_{111}^{*}+q^{*}\varphi_{121}^{*}+r^{*}\varphi_{120}^{*},\\
\varphi_{120x}^{*}&=-i\lambda_{1}^{*}\varphi_{120}^{*}-q\varphi_{110}^{*},\\
\varphi_{121x}^{*}&=-i\lambda_{1}^{*}\varphi_{121}^{*}-q\varphi_{111}^{*}-r\varphi_{110}^{*},\\
\end{split}
\end{equation}
and thus we can get
\begin{equation}
\varphi_{210}=-\varphi_{120}^{*}, \  \ \varphi_{211}=-\varphi_{121}^{*}, \  \
\varphi_{220}=\varphi_{110}^{*}, \  \  \varphi_{221}=\varphi_{111}^{*},
\end{equation}
and $\lambda_{2}=\lambda_{1}^{*}$.
\end{proof}

According to the Lemma \ref{linearconstraint}, by choosing $a_{10}=a_{11}=1$, $d_{10}=d_{11}=1$,  we can get the one fold Darboux transformation for Frobenius NLS equation as follows
\begin{equation}
\begin{split}
T=T_{1}(\lambda)
 =\left(\begin{matrix}
(\lambda+a_{0})E+(\lambda+a_{1})\Gamma&b_{0}E+b_{1}\Gamma\\
c_{0}E+c_{1}\Gamma&(\lambda+d_{0})E+(\lambda+d_{1})\Gamma\\
\end{matrix}
\right)\\
\\=\left(\begin{matrix}
E+\Gamma&0\\
0&E+\Gamma\\
\end{matrix}
\right)\lambda+\left(\begin{matrix}
a_{0}E+a_{1}\Gamma&b_{0}E+b_{1}\Gamma\\
c_{0}E+c_{1}\Gamma&d_{0}E+d_{1}\Gamma\\
\end{matrix}
\right),\\
\end{split}
\end{equation}
and then the new solution $q^{[1]}, r^{[1]}$ can be obtained by the Darboux transformation
\begin{equation} \label{db}
\begin{cases}
q^{[1]}=q+2ib_{0},\\
r^{[1]}=r+2i(b_{1}-b_{0}).
\end{cases}
\end{equation}

The Darboux transformation matrix $ T_{1} $ must satisfy the following equation
\begin{equation}
T_{1}(\lambda,\lambda_{j})\varphi_{j}=0,j=1,2.
\end{equation}

From the above equation, we can get the  expressions of functions $ b_{0},b_{1}$ easily by the eigenfunctions as follows

$b_{0}=\frac{(\lambda_{1}-\lambda_{2})\varphi_{110}\varphi_{210}}
{\varphi_{110}\varphi_{220}-\varphi_{120}\varphi_{210}}$,\\

$b_{1}=\frac{(\lambda_{1}-\lambda_{2})(\varphi_{110}^2\varphi_{211}\varphi_{220}+\varphi_{110}\varphi_{120}\varphi_{210}\varphi_{221}
+\varphi_{110}\varphi_{121}\varphi_{210}^2-\varphi_{110}\varphi_{120}\varphi_{210}\varphi_{211}
-\varphi_{111}\varphi_{120}\varphi_{210}^2-\varphi_{110}^2\varphi_{210}\varphi_{221})}{(\varphi_{110}\varphi_{220}-\varphi_{120}\varphi_{210})^2}$.\\

Combining with $\eqref{db}$, the following theorem can be derived.
\begin{theorem}
The Darboux transformation of the Frobenius NLS equation can be as
\begin{equation}\label{3.04}
\begin{cases}
q^{[1]}=q+2i\frac{(\lambda_{1}-\lambda_{2})\varphi_{110}\varphi_{210}}{\varphi_{110}\varphi_{220}-\varphi_{120}\varphi_{210}},\\
r^{[1]}=r+2i\frac{(\lambda_{1}-\lambda_{2})(\varphi_{110}^{2}(\varphi_{220}(\varphi_{211}-\varphi_{210})-\varphi_{210}\varphi_{221})
+\varphi_{210}^2(\varphi_{110}(\varphi_{120}+\varphi_{121})-\varphi_{111}\varphi_{120})
+\varphi_{110}\varphi_{120}\varphi_{210}(\varphi_{221}-\varphi_{211}))}
{(\varphi_{110}\varphi_{220}-\varphi_{120}\varphi_{210})^2}.
\end{cases}
\end{equation}
\end{theorem}

Such transformation of the Frobenius NLS equation will help us to find soliton solutions in the next section.

\section{The soliton solution  of the Frobenius NLS equation}
\   \
In this section, our aims are to construct the soliton solutions of the Frobenius NLS equation by assuming suitable seed solutions. For simplicity, we assume a trivial solution as $q(x,t)=r(x,t)=0$, we can get soliton solutions of the Frobenius NLS equation by the Darboux transformation in eq.\eqref{3.04}. Then we can get the following solution by solving the linear equation $\eqref{laxfnls}$
\begin{equation}\notag
\begin{split}
\varphi_{j10}=\varphi_{j11}=e^{-i\lambda_{j}x-2i\lambda_{j}^{2}t},\\
\varphi_{j20}=\varphi_{j21}=e^{i\lambda_{j}x+2i\lambda_{j}^{2}t}.
\end{split}
\end{equation}

Let $j=1,2$, $\lambda_{1}=\alpha+i\beta$, substituting $\varphi_{j10}$, $\varphi_{j11}$, $\varphi_{j20}$, $\varphi_{j21}$ into $\eqref{3.04}$, then the new solution becomes
 \begin{equation}
\begin{cases}
q^{[1]}=2\beta e^{-2i(2\alpha^2t-2\beta^2t+\alpha x)}sech(2\beta x+8\alpha\beta t),\\
r^{[1]}=\beta(e^{2i(2\beta^2t-2\alpha^2t-\alpha x)+8\alpha\beta t+2\beta x}-1)sech^{2}(2\beta x+8\alpha\beta t).
\end{cases}
 \end{equation}

If we substitute these eigenfunctions into the Darboux transformation equations $\eqref{3.04}$  with $\alpha=0.5$, $\beta=1$,  then we can obtain the  soliton solutions of the Frobenius NLS equation whose graphs are shown in Fig.1.

 \section{Breather solutions of the Frobenius NLS equation}

From the previous section, solitons have been discussed for the Frobenius NLS equation. In this section, we will focus on a new kind of solution which  start from periodic seed solutions by the Darboux transformation. The new solutions are named breather solutions. The periodic seed solutions can be assumed as $q=ce^{i\rho}$, $r=ide^{i\rho}$, $\rho=ax+bt$, $a$, $b$, $c$, $d$ are arbitrary real constants. Substituting the periodic seed solutions into \eqref{Z2}, through proper simplification, we can obtain a constraint

 \begin{equation}
b+a^{2}-2c^{2}=0.
 \end{equation}
By defining
\begin{equation}
R=\sqrt{-a^2-4a\lambda-4c^2-4\lambda^2},
\end{equation}
the following basic solutions $\psi_{110}$, $\psi_{111}$, $\psi_{120}$, $\psi_{121}$ are obtained in terms of $R$ from the the spectral problem of eq.\eqref{laxfnls} as
\begin{equation} \notag
\begin{split}
\psi_{110}(R)&=n_{1}e^{\frac{i}{2}\rho+R(\frac{1}{2}x+(\lambda-\frac{a}{2})t)},\\
\psi_{111}(R)&=k_{1}e^{\frac{i}{2}\rho+R(\frac{1}{2}x+(\lambda-\frac{a}{2})t)},\\
\psi_{120}(R)&=\frac{(ai+2i\lambda+R)n_{1}}{2c}e^{-\frac{i}{2}\rho+R(\frac{1}{2}x+(\lambda-\frac{a}{2})t)},\\
\psi_{121}(R)&=\frac{(aci+2ci\lambda+Rc)k_{1}+(-iRd+ad+2d\lambda)n1}{2c^2}e^{-\frac{i}{2}\rho+R(\frac{1}{2}x+(\lambda-\frac{a}{2})t)}.\\
\end{split}
\end{equation}

The solutions will be trivial, if we put  eigenfunction  $\psi_{110}$, $\psi_{120}$, $\psi_{111}$, $\psi_{121}$, into the Darboux transformation directly. So, according to the principle of linear superposition and contract conditions, we can construct the solution of more complicated but very meaningful. The eigenfunctions $\varphi_{110}$, $\varphi_{120}$, $\varphi_{111}$, $\varphi_{121}$, associated with $\lambda_{1}$ can be given as
 \begin{equation}\label{22}
\begin{split} \notag
\varphi_{110}=\psi_{110}(R)(\lambda_{1})+\psi_{110}(-R)(\lambda_{1})-\psi_{120}^{*}(R)(\lambda_{1}^{*})-\psi_{120}^{*}(-R)(\lambda_{1}^{*}),\\
\varphi_{111}=\psi_{111}(R)(\lambda_{1})+\psi_{111}(-R)(\lambda_{1})-\psi_{121}^{*}(R)(\lambda_{1}^{*})-\psi_{121}^{*}(-R)(\lambda_{1}^{*}),\\
\varphi_{120}=\psi_{120}(R)(\lambda_{1})+\psi_{120}(-R)(\lambda_{1})+\psi_{110}^{*}(R)(\lambda_{1}^{*})+\psi_{110}^{*}(-R)(\lambda_{1}^{*}),\\
\varphi_{121}=\psi_{121}(R)(\lambda_{1})+\psi_{121}(-R)(\lambda_{1})+\psi_{111}^{*}(R)(\lambda_{1}^{*})+\psi_{111}^{*}(-R)(\lambda_{1}^{*}).\\
\end{split}
\end{equation}

It can be proved that $\varphi_{110}$, $\varphi_{111}$, $\varphi_{120}$ and $\varphi_{121}$ are also the solution of the Lax equation with $\lambda_{1}=\alpha+i\beta,$ $\alpha=-\frac{a}{2}$. Putting $\varphi_{110}$, $\varphi_{111}$, $\varphi_{120}$ and $\varphi_{121}$ into the Darboux transformation will lead to the construction of breather solutions of the Frobenius NLS equation. For brevity, in the following, an explicit expressions of breather solutions $q_{b}^{[1]}$, $r_{b}^{[1]}$  are obtained with specific parameters $n_{1}=k_{1}=d=1$
\begin{equation}\label{hx}
\begin{cases}
q_{b}^{[1]}=\frac{e^{i\rho}(e^{-4iR\beta t}(c^2+2R\beta-2\beta^2)+(e^{-4Rat+2Rx}+e^{4Rat-2Rx})\beta c+e^{4iR\beta t}(c^2-2R\beta -2\beta^2)}
{2c \cos(4R\beta t)-\beta(e^{-2Rx+4Rat}+e^{2Rx-4Rat})},\\
r_{b}^{[1]}=\frac{v_{1}}
{2(\beta c^3 -c^4)\cos(4R\beta t)^2+4(\beta c^3-\beta^2c^2)\cos(4R\beta t)cosh(2R(2at-x))
+(\beta^3c-\beta^2c^2)(cosh(4R(2at-x))+1)}.\\
\end{cases}
\end{equation}

Where $v_{1}$  is explicitly given in Appendix A and $R=\sqrt{\beta^{2}-c^{2}}$.
\begin{remark}
From the graphs of the breather solutions of the Frobenius NLS equation which are plotted in Fig. 2. We observe that the breather solution $q_{b}^{[1]}$ is periodic, while the amplitude of the optical wave of the breather solution $r_{b}^{[1]}$  is not periodic any more because of different amplitudes of peaks from the Fig. 2. We can also observe that there exist period-like fluctuations in the background  of fiber wave $r_{b}^{[1]}$ in the Fig. 3, Fig. 4. One of the important reasons causing the above phenomena is that $r_{b}^{[1]}$ is affected by $q_{b}^{[1]}$ which can be seen clearly from the equation \eqref{fnlsequation} of the Frobenius NLS equation.
\end{remark}

\section{Rouge waves solutions of the Frobenius NLS equation}
\  \
If the eigenvalue $\lambda$ of eigenfunction $\varphi$ is a null point, a limit as a rogue wave will happen in the wave propagation direction. Here, let's clarify this statement. We can find  $R=0$, when $\lambda=-\frac{a}{2}+ic$ , i.e. $\alpha=-\frac{a}{2}$, $\beta=c$. Then we do the Taylor expansion to the breather solution \eqref{hx} at $\beta=c$, $\alpha=-\frac{a}{2}$, one rogue wave solution of $q_{w}$, $r_{w}$ will be obtained. The form will be given in the following
\begin{equation}\label{gb}
\begin{split}
\begin{cases}
q_{w}=\frac{c(-64c^{2}\alpha^{2}t^{2}-16c^{4}t^{2}-32c^{2}\alpha xt+16ic^{2}t-4c^{2}x^{2}+3)e^{-2i(2\alpha^{2}t-c^{2}t+\alpha x)}}{64\alpha^{2}c^{2}t^{2}+16c^{4}t^{2}+32\alpha c^{2}tx+4c^{2}x^{2}+1},\\
r_{w}=\frac{e^{-i(a^2t-2c^2t-ax)}v_{2}+v_{3}}{(256t^4(a^2+c^2)^2+c^4(16x^3(x-8at)+t^2(32+384a^2x^2+128x((x-4at)c^2-4a^3t)))
+(8(x-2at)^2)c^2+1},
\end{cases}
\end{split}
\end{equation}
where $v_{2}$, $v_{3}$ are defined in Appendix B.
\begin{remark}
From the density graph of the rouge wave solutions $q_{w}$, $r_{w}$,  we can see that the rouge wave solution $q_{w}$ have a single peak with two caves on both sides of the peak, while the rouge wave solution $r_{w}$ only have a single peak without any obvious caves around. Because the optical pulse $r_{w}$ is affected by another optical pulse $q_{w}$, the optical signal $r_{w}$ has fluctuations around the rouge wave (see in Fig.5 and Fig.6). Similarly as $q_{w}$, the  optical pulse  $r_{w}$ only exists locally with all variables and disappear  as time and space go far which can be seen in the Fig. 7.
\end{remark}

\  \\
\   \

\   \

\   \
\   \
\begin{center}
Appendix A: Explicit expressions for $v_{1}$
\end{center}
\begin{footnotesize}
\begin{quote}
$v_{1}=\frac{1}{4}(8ic^2\beta^3-8ic^3\beta^2+32c\beta^3-16c^2\beta^2-16c^3\beta -(6\beta^2c^2+2\beta c^3-8\beta^3c-4i\beta^3c^2+4i\beta^2c^3)\cosh(4R(2at-x))+(8c^4-8\beta c^3)\cos(4R\beta t)^2)\sin(\rho)
+(2i\beta^3c^2-2i\beta^2c^3+6R\beta^2c-5R\beta c^2+6\beta^3c-\beta c^3-5\beta^2c^2+2i\beta^2c^2R-2i\beta c^3R)\sin(8Rt\beta+\rho)
+(2i\beta^3c^2-2i\beta^2c^3+6\beta^3c-\beta c^3-5\beta^2c^2+2iR\beta c^3-6R\beta^2c+5R\beta c^2-2i\beta^2c^2R)\sin(-8Rt\beta+\rho)
+\sinh(2R(2at-x))((2\beta^3c-4R\beta^3-4iR\beta^2c^2-4i\beta^3c^2+4i\beta^2c^3+4\beta^4-2\beta c^3+4i\beta^4c+2R\beta^2c+4\beta^2c^2++4i\beta c^4)\sin(4R\beta t+\rho)
+(2\beta c^3-4\beta^2c^2-4R\beta^3+4iR\beta^3c-2\beta^3c+4\beta^4-4i\beta^2c^2R+4i\beta^3c^2+4i\beta^2c^3-4i\beta c^4+2\beta^2cR-4i\beta^4c)\sin(-4R\beta t+\rho)
+(4\beta^3c-4\beta c^3)\sin(4R\beta t))-2R\beta c^2\sin(8R\beta t)
+((4iR\beta c^3-2R\beta c^2+4R\beta^2c-4i\beta^2c^2R)\sin(\rho)+2i\beta c^2R
+((2i\beta^3c+4i\beta^2c^2-4i\beta^4-4\beta c^4-2i\beta c^3+4\beta^2c^3-2iR\beta^2c)\cos(-4R\beta t+\rho)
+(-2i\beta^3c-4i\beta^2c^2+4i\beta^4-4\beta^3c^2+2i\beta c^3-4\beta^2c^3-2iR\beta^2c)\cos(4R\beta t+\rho))\sinh(2R(2at-x))
+(20R\beta^3+10\beta c^3+10\beta^3c+4iR\beta^3c-4iR\beta^2c^2+4i\beta c^4+4i\beta^3c^2-4i\beta^4c-10\beta^2cR-20\beta^4-8R\beta c^2+4i\beta^2c^3)\cosh(2R(2at-x))\sin(-4R\beta t+\rho)
+(4i\beta^3c^2+4i\beta c^4-4i\beta^2c^3+10R\beta^2c-20\beta^4+10\beta c^3+10\beta^3c-20R\beta^3-4i\beta^4c-4iR\beta^2c^2+8R\beta c^2)\cosh(2R(2at-x))\sin(4R\beta t+\rho)
+(4\beta c^4+4\beta^4c+4R\beta^3c+4iR\beta^3-4R\beta^2c^2)\sinh(2R(2at-x))\cos(4R\beta t+\rho)
+((4R\beta c^3-4R\beta^2c^2+2iR\beta c^2-4i\beta^2cR)\cos(\rho)
-4i\beta^2cR\cos(4R\beta t)))\sinh(4R(2at-x))
+(2R\beta c^3+i\beta c^3+2\beta^3c^2-2\beta^2c^3+5i\beta^2c^2-6i\beta^3c-5iR\beta c^2-2R\beta^2c^2+6i\beta^2cR)\cos(-8R\beta t+\rho)
+(2\beta^3c^2-2\beta^2c^3+5i\beta c^2R+5i\beta^2c^2-6i\beta^3c-6iR\beta^2c-2R\beta c^3)\cos(8R\beta t+\rho)
-(4\beta^2c^3-10i\beta^3c-8iR\beta c^2)\cos(4R\beta t+\rho)
+((2i\beta c^3-4\beta^2c^3+6i\beta^2c^2-8i\beta^3c+4\beta^3c^2)\cosh(4R(2at-x))
+4R\beta^2c\cosh(2R(2at-x))\sin(4R\beta t)
+(8i\beta c^3-8ic^4)\cos(4R\beta t)^2+4\beta^2c^3-4\beta c^4+2i\beta c^3-2i\beta^2c^2+
+8\beta^3c^2-8\beta^2c^3+16i\beta^2c^2+16i\beta c^3-32i\beta^3c)\cos(\rho)
+(2i\beta c^3-2i\beta^2c^2+2R\beta^2c^2+i\beta c^3+4\beta^2c^3-4\beta c^4)\cos(8R\beta t)
+(((16i\beta c^3-16i\beta^2c^2)\cos(\rho)+16\beta^2c^3-16\beta^3c^2-4i\beta c^3-8i\beta^2c^2+12i\beta^3c)\cos(4R\beta t)
+(8iR\beta c^2+10iR\beta^2c20i\beta^4+4\beta c^4-4\beta^2c^3+4\beta^3c^2-4\beta^4c
+4R\beta^3c-4R\beta^2c^2-20iR\beta^3-10i\beta c^3-10i\beta^3c
+4\beta^3c^2-4\beta^4c+4\beta c^4-10i\beta c^3+4R\beta^2c^2-4R\beta^3c+20iR\beta^3)+((16\beta^2c^2-16\beta c^3)\cos(4R\beta t)\sin(\rho)+20i\beta^4-10icR\beta^2))\cosh(2R(2at-x))
+(4\beta^3c^2-4\beta^4c+4R\beta^3c-4R\beta^2c^2+4i\beta^3R)\sinh(2R(2at-x)))\cos(-4R\beta t+\rho)
+4i\beta^3c+8\beta^4c-8\beta^3c^2+4i\beta^2c^2-8i\beta^4\\$
\end{quote}
\end{footnotesize}
\begin{center}

Appendix B: Expressions for $v_{2}$, $v_{3}$,
\end{center}
\begin{quote}
$v_{2}=-48ia^2c^3t^2x+24iac^3tx^2+24iac^2tx+32a^2c^4t^3+8c^4tx^2+16ac^3t^2-8c^3tx+24ic^4t^2-4ic^3x^3-6ic^2x^2+icx+32c^6t^3+18c^2t-16ic^5t^2x-24ia^2c^2t^2-
2iact-32ac^4t^2x+32ia^3c^3t^3+32iac^5t^3-5i,\\
v_{3}=16c^5t^2(ix-2ct-2iat)+c^4(8t^2(i+4ax-4a^2t)-8tx^2)+c^3(4x(ix^2+2t(1-3iax)+16at^2(ia(3x-2at)-1))
+2c^2(-i(4a^2t^2-8atx+x^2)-t)+ic(2at-x).\\$
\end{quote}

{\bf {Acknowledgements:}}
Chuanzhong Li  is  supported by the National Natural Science Foundation of China under Grant No. 11571192 and K. C. Wong Magna Fund in
Ningbo University.


\begin{figure}[ht]
\centering
$(|q^{[1]}|^2)$
\includegraphics[width=5cm]{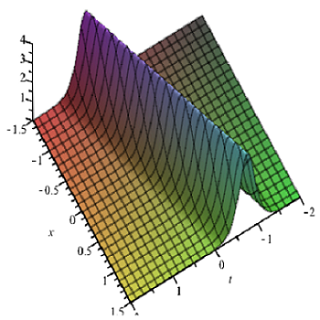}
$(|r^{[1]}|^2)$
 \includegraphics[width=5.5cm]{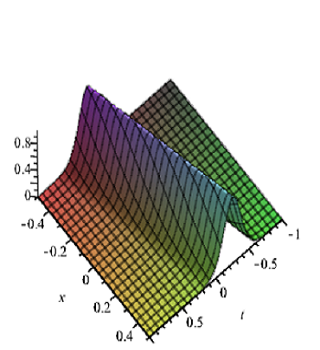}
\caption{Soliton solutions($q^{[1]}, r^{[1]}$) of the Frobenius NLS equation with $\alpha=0.5,\beta=1$.}
\end{figure}

\begin{figure}[ht]
\centering
$(|q_{b}^{[1]}|^2)$
\includegraphics[width=4.5cm]{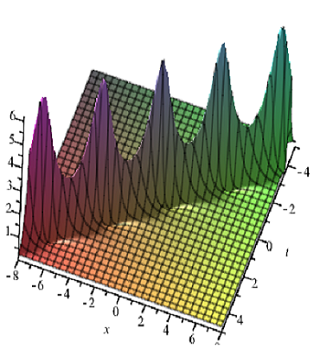}
$(|r_{b}^{[1]}|^2)$
\includegraphics[width=5cm]{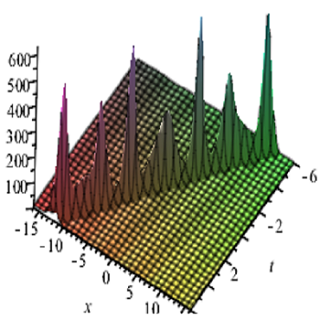}
\caption{Breather solution ($q_{b}^{[1]}, r_{b}^{[1]}$) of the Frobenius NLS equation with $b=-0.5, c=0.5, d=1,\alpha=0.5,\beta=1$.}
\end{figure}

\begin{figure}[ht]
\centering
$(|q_{b}^{[1]}|^2)$
\includegraphics[width=5cm]{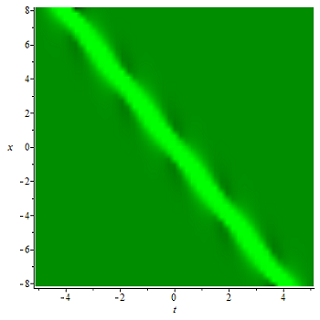}
$(|r_{b}^{[1]}|^2)$
\includegraphics[width=5cm]{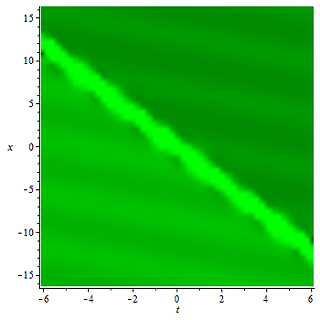}
\caption{Breather solution ($q_{b}^{[1]}, r_{b}^{[1]}$) of the Frobenius NLS equation with $b=-0.5, c=0.5, d=1, \alpha=0.5,\beta=1$.}
\end{figure}

\begin{figure}[ht]
\centering
($|r_{b}^{[1]}|^2$ with $x=0$)
\includegraphics[width=4.5cm]{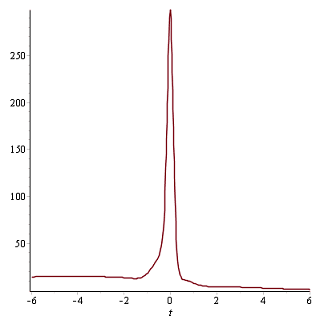}
($|r_{b}^{[1]}|^2$ with $t=0$)
\includegraphics[width=4.5cm]{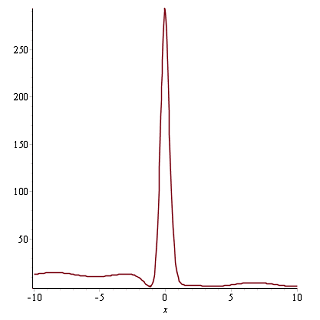}
\caption{Breather solution $r_{b}^{[1]}$ of the Frobenius NLS equation with $b=-0.5, c=0.5, d=1,\alpha=0.5,\beta=1$.}
\end{figure}

\begin{figure}[ht]
\centering
$(|q_{w}|^2)$
\includegraphics[width=6cm]{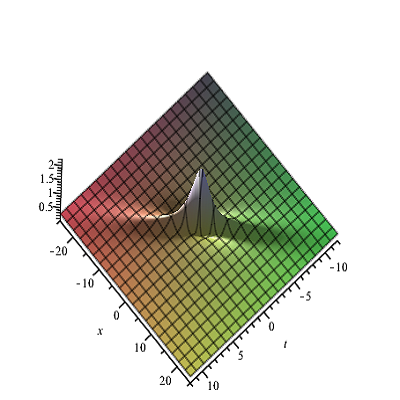}
$(|r_{w}|^2)$
\includegraphics[width=6cm]{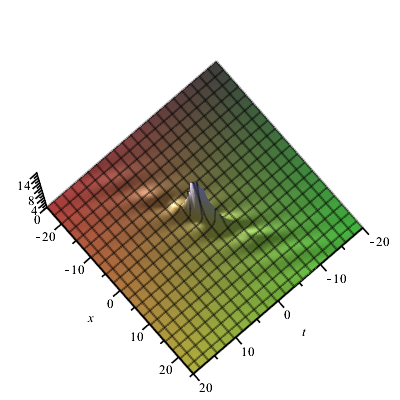}
\caption{Rogue wave solution ($q_{w}, r_{w}$) of the Frobenius NLS equation with $b=-0.5, c=0.5, d=1,\alpha=0.5,\beta=1$.}
\end{figure}

\begin{figure}[ht]
\centering
$(|q_{w}|^2)$
\includegraphics[width=5cm]{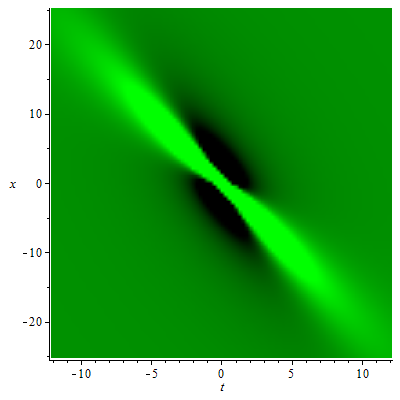}
$(|r_{w}|^2)$
\includegraphics[width=5cm]{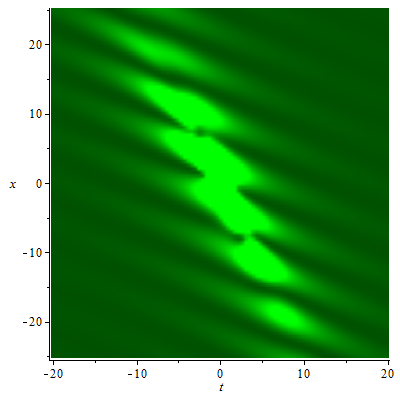}
\caption{Rogue wave solution ($q_{w}, r_{w}$) of the Frobenius NLS equation with $b=-0.5, c=0.5, d=1, \alpha=0.5,\beta=1$. }
\label{densityrogue}
\end{figure}

\begin{figure}[ht]
\centering
($|r_{w}|^2$ with $x=0$)
\includegraphics[width=4.5cm]{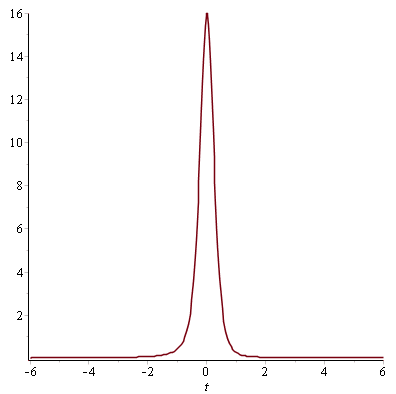}
($|r_{w}|^2$ with $t=0$)
\includegraphics[width=4.5cm]{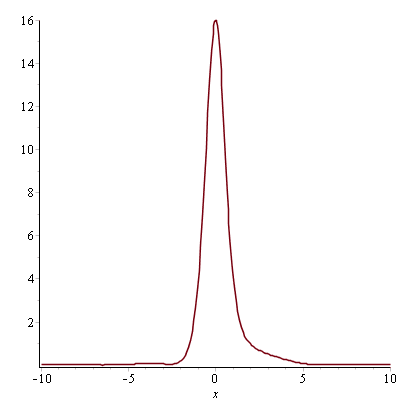}
\caption{Rogue wave  solution $r_{w}$ of the Frobenius NLS equation with $b=-0.5, c=0.5, d=1,\alpha=0.5,\beta=1$.}
\end{figure}

\end{document}